\documentclass[reqno,a4paper,final,11pt,titlepage]{amsart}
\usepackage{geometry}
\renewcommand{\paragraph}{\subsubsection*}
\usepackage[UKenglish]{babel}
\usepackage{microtype}
\usepackage{booktabs}

\usepackage{lineno}
\usepackage{url}
\usepackage[
  pdftex,
  bookmarks=true,%
  bookmarksnumbered=true,%
  hypertexnames=false,%
  breaklinks=true,%
  linkbordercolor={0 0 1}%
]{hyperref}

\providecommand{\urlstyle}[1]{}
\providecommand{\doi}[1]{\href{http://dx.doi.org/#1}{\nolinkurl{doi:#1}}}

\usepackage{float}
\usepackage{xcolor}
\usepackage{amsmath}
\usepackage{amsfonts}
\usepackage{amsthm}
\usepackage{mathtools}
\usepackage{bm}
\usepackage{todonotes}
 
\usepackage{cleveref}
\usepackage{caption}
\usepackage{subcaption}
\usepackage{thmtools} 
\usepackage{thm-restate}

\newtheorem{theorem}{Theorem}

\newtheorem{proposition}[theorem]{Proposition}

\newtheorem{lemma}[theorem]{Lemma}

\newcommand{\N}{\mathbb{N}}
\newcommand{\Z}{\mathbb{Z}}
\newcommand{\Q}{\mathbb{Q}}
\newcommand{\R}{\mathbb{R}}

\newcommand{\C}{\mathbb{C}}
\newcommand{\K}{\mathbb{K}}

 \renewcommand{\vec}[1]{{\bm #1}}

\newcommand{\lcm}{\mathrm{lcm}}

\newcommand{\NP}{\bf{NP}}

\newcommand{\PSPACE}{\bf{PSPACE}}

\newcommand{\AM}{\bf{AM}}

\newcommand{\sharpP}{\mbox{{\#\bf P}}}

 \newcommand{\HN}{\mathrm{HN}}
 \newcommand{\HNZ}{\mathrm{HNP}}
 \newcommand{\DIM}{\mathrm{DIM}}

\begin{document}

\title{A parametric version of the Hilbert Nullstellensatz}

\author[R.~Ait El Manssour]{Rida Ait El Manssour}
\address{Rida Ait El Manssour, Department of Computer Science, University of Oxford, Oxford, UK}
\email{rida.aitelmanssour@cs.ox.ac.uk}

\author[N.~Balaji]{Nikhil Balaji}
\address{Nikhil Balaji, IIT Delhi, New Delhi, India}
\email{nbalaji@cse.iitd.ac.in}

\author[K.~Nosan]{Klara Nosan}
\address{Klara Nosan, Universit\'e Paris Cité, CNRS, IRIF, Paris, France}
\email{nosan@irif.fr}

\author[M.~Shirmohammadi]{Mahsa Shirmohammadi}
\address{Mahsa Shirmohammadi, Universit\'e Paris Cité, CNRS, IRIF, Paris, France}
\email{mahsa@irif.fr}

\author[J.~Worrell]{James Worrell}
\address{James Worrell, Department of Computer Science, University of Oxford, Oxford, UK}
\email{jbw@cs.ox.ac.uk}

\begin{abstract}
 Hilbert's Nullstellensatz  is a fundamental result in algebraic geometry that gives a necessary and sufficient condition for 
 a finite collection of multivariate polynomials to have a common zero in an algebraically closed field.
 Associated with this result, there is the computational problem $\HN$  
of determining whether a system of polynomials with coefficients in the field of rational numbers has a common zero over the field of algebraic numbers.   

In an influential paper, Koiran showed that $\HN$ 
can be determined in the polynomial hierarchy assuming the Generalised Riemann Hypothesis (GRH).  More precisely, he showed that $\HN$ lies in the complexity class ${\AM}$ under GRH.  In a later work he generalised this result by showing that 
the problem $\DIM$, which asks to determine 
the dimension of the set of solutions of a
given
polynomial system, also lies in ${\AM}$
subject to GRH.

In this paper we study the solvability of polynomial equations over arbitrary algebraically closed fields of characteristic zero.  Up to isomorphism, 
every such field is the algebraic closure of a field of
rational functions.  We thus formulate a parametric version of $\HN$, called $\HNZ$, in which the input is a system of polynomials with coefficients
in a function field $\Q(\vec{x})$ and the task is to determine whether
the polynomials have a common zero in the algebraic closure $\overline{\Q(\vec{x})}$.

We observe that
Koiran's proof that $\DIM$ lies in {\AM} can be interpreted as  
a randomised polynomial-time reduction 
of $\DIM$ to $\HNZ$, followed by an argument that $\HNZ$
lies in ${\AM}$.  Our main contribution is a self-contained proof 
that $\HNZ$ lies in ${\bf AM}$ that follows the same basic idea as Koiran's argument -- namely random instantiation of the parameters --  but whose justification is purely algebraic, relying on a parametric version of Hilbert's Nullstellensatz, and avoiding 
recourse to semi-algebraic geometry.

\end{abstract}

\maketitle

\section{Introduction}

Hilbert's Nullstellensatz is a fundamental result in algebraic geometry. 
Given a system of polynomial equations 
\begin{equation}
\label{eq:system-hn}
f_1(y_1,\ldots,y_n) = 0,\; \ldots, \; f_\ell(y_1,\ldots,y_n) = 0
\end{equation}
where $f_i \in k[y_1, \dots, y_n]$ for $k$ a field with algebraic closure $K$, the Nullstellensatz says that the system is unsatisfiable over $K$ if and only if there exist polynomials $g_1, \dots, g_\ell \in k[y_1, \dots, y_n]$ such that
\begin{equation} \label{eqn:ideal}
\sum_{i=1}^\ell f_i g_i = 1
\end{equation}

A naturally associated computational problem asks to 
determine whether a given family of polynomials $f_1, \dots, f_\ell$ has a common zero.   
More specifically, we have the problem
$\HN$ in which the input is system of polynomial equations with rational coefficients and the task is to determine
whether there exists a solution in the field of algebraic numbers~$\overline{\Q}$.\footnote{By the Nullstellensatz this is equivalent to asking whether a system of polynomial equations with rational coefficients admits a common solution over $\C$.}  

The characterisation given in Equation~\eqref{eqn:ideal} essentially reduces $\HN$ to an \emph{ideal membership problem}, namely whether the constant 1 lies in the ideal generated by the $f_i$.  By a result of Mayr and Meyer~\cite{mayr-meyer}, this places $\HN$ in the complexity class {\bf EXPSPACE}. On the other hand, it is easy to reduce the Boolean satisfiability problem to $\HN$, making it at least {\NP}-hard.
 
 Over the years, various \emph{Effective Nullstellens{\"a}tze}~\cite{kollar, krick01, jel-effective} have yielded single-exponential degree bounds on $g_i$, reduces the ideal membership problem arising from the Nullstellensatz to that of solving a single-exponential-sized system of linear equations. This yields a {\PSPACE} upper bound for $\HN$. In an influential paper~\cite{KOIRAN1996273, KOIRAN-technical}, Koiran proved that $\HN$ lies in the complexity class {\AM} (and thereby in the second level of the polynomial hierarchy) if one assumes GRH. He achieves this bound by exhibiting a \emph{gap theorem}: For any system $\mathcal{S}$ of polynomial equations with integer coefficients, if $\mathcal S$ is satisfiable over $\overline{\Q}$ then there are at least~$B$ primes $p$ such that $\mathcal{S}$ is satisfiable over 
 the finite field
 $\mathbb{F}_p$, whereas if $\mathcal S$
 is not satisfiable over $\overline{\Q}$ then there are 
 at most~$A$ primes for which $\mathcal{S}$ is satisfiable over $\mathbb{F}_p$. Here $A < B$ are both numbers whose magnitude is a single-exponential function of the parameters (namely, number of variables, degree, bit size of coefficients, number of polynomial equations) of the system $\mathcal{S}$. Therefore to decide $\HN$ it suffices to count the number of primes for which $\mathcal{S}$ is satisfiable over~$\mathbb{F}_p$.  Since 
 satisfiability over~$\mathbb{F}_p$
 has an easy~{\NP} procedure, one obtains a $\sharpP^{\NP}$ algorithm for $\HN$. By further observing that there is a sufficiently large gap between $A$ and~$B$, Koiran proves that it suffices to approximately count the number of primes, which yields the claimed upper bound for the problem.

In this paper we study a version of $\HN$ over arbitrary algebraically closed fields 
of characteristic zero.  In general, such fields have the form $\overline{\Q(\vec{x})}$ for 
$\vec{x}$ a set of indeterminates.  Specifically we consider the problem $\HNZ$, which asks to determine 
whether a system of multivariate polynomials with coefficients in $\Q(\vec{x})$ has a solution
in $\overline{\Q(\vec{x})}$, where $\vec{x}=(x_1,\ldots,x_m)$.
Evidently $\HNZ$ generalises $\HN$.  We remark also that
systems of polynomial equations with coefficients in $\Q(\vec{x})$ are
central objects of study in algebraic combinatorics and in the theory of formal languages, where they are used to specify generating functions of combinatorial objects (see, e.g., the overviews in~\cite{BanderierD15,KuichSalomaa}).

We present a generalisation of the technique introduced in \cite{KOIRAN1996273}, allowing the use of algebraic arguments to establish a randomised polynomial-time reduction from $\HNZ$
to $\HN$.  In particular, given a system of polynomials with coefficients in $\Q(\vec{x})$, we compute a bound $D$ such that if we specialise the variables $x_i$ uniformly at randomly from $\{1,\ldots,D\}$ and use Koiran's algorithm to determine whether the specialised system is satisfiable over $\overline{\Q}$, with high enough probability, the algorithm will give a correct answer to the problem over $\overline{\Q(\vec{x})}$. We thus obtain our main result:

\begin{restatable}{theorem}{thmmain}
\label{thm:main}
$\HNZ$ is in $\AM$ assuming GRH. 
\end{restatable}

To further undertand the context of Theorem~\ref{thm:main}, it is useful to take a geometric view.
The set of common zeros of a collection of polynomials $f_1,\ldots,f_k \in \Q[y_1,\ldots,y_n]$ is an 
\emph{algebraic variety} $V \subseteq \overline{\Q}^n$.  The $\HN$ problem asks whether this variety is non-empty.  More generally, we have the 
$\DIM$ problem which, given an integer $d$, asks whether $V$ has dimension at least $d$.  Note that $\HN$ is precisely the problem $\DIM$ specialised to $d=0$.  A randomised polynomial-time reduction of $\DIM$ to $\HN$ is established in~\cite{koiran-dim}, thus placing $\DIM$ in {\AM} assuming GRH.
The reduction works by first applying a random linear transformation~$A$ to the variety~$V$ such that if $V$ has dimension at least $d$ then with high probability the projection onto the coordinates
$y_1,\ldots,y_d$
of the image $AV$ is Zariski-dense in $\overline{\Q}^d$. The second step of the reduction involves randomly choosing  integers $\alpha_1,\ldots,\alpha_d$, adding equations 
$y_1=\alpha_1,\ldots,y_d =\alpha_d$
to the system defining~$AV$, and verifying the satisfiability of the new system over~$\overline{\Q}$ via the $\HN$ algorithm.

The $\HNZ$ problem has a 
geometric interpretation closed related to the above discussion of $\DIM$.  
Given polynomials $f_1,\ldots,f_k \in \Q[\vec{x}][y_1,\ldots,y_n]$ where $\vec{x} := (x_1,\ldots,x_m)$, let $V$ be the subvariety of $\overline{\Q}^{m+n}$  defined by the system when treating the parameters $x_1,\ldots,x_m$ as variables. Then the satisfiability of the system
 over $\overline{\Q(\vec{x})}$  is equivalent to the projection of $V$
onto its first $m$ coordinates being Zariski-dense in $\overline{\Q}^m$. Indeed, by the weak Hilbert Nullstellensatz for $\overline{\Q(\vec{x})}$, 
the system if satisfiable over 
$\overline{\Q(\vec{x})}$
iff
$1\not\in I$ for 
$I \subseteq \Q(\vec{x})[\vec{y}]$ the ideal
generated by $f_1,\ldots,f_k$.
Since the constants in this setting belong to $\Q(\vec{x})$, the condition corresponds to $I \cap \Q[x_1,\ldots,x_m]$ being zero. By \cite[Section~9.5, Corollary~4 and Proposition~5]{cox2013ideals} this is equivalent to the projection of
$V$ on $x_1,\ldots,x_m$ being dense in $\overline{\Q}^m$.

The above discussion shows that we can
see $\HNZ$ as a midpoint in the reduction 
of $\DIM$ to $\HN$---the first step 
being a randomised polynomial-time
reduction of $\DIM$ to $\HNZ$,
and the second step being a randomised polynomial-time reduction of $\HNZ$ to $\HN$.
However the approach to the $\HNZ$ problem taken in the present paper 
differs from~\cite{koiran-dim} in its correctness proof and error analysis.
Our probability bounds use the Polynomial Identity Lemma (also called the DeMillo-Lipton-Schwartz-Zippel Lemma), in support of which we use a parametric version of the Nullstellensatz and basic facts about ring and field extensions.  By contrast, the approach of~\cite{koiran-dim} relies on real algebraic geometry via an identification of sub-varieties of $\C^n$ with definable sets in $\R^{2n}$.  In particular, it uses a technique of~\cite{Koiran95-volume} for counting integer points in semi-algebraic sets in terms of their volume and number of connected components.

Finally, we mention that HNP can be formulated using the so-called \emph{generic quantifier} $\forall^*$ considered by Koiran~\cite{koiran-dim}, where $\forall^* \boldsymbol x \, \varphi(\boldsymbol x)$ expresses that the that the predicate $\varphi(\boldsymbol x)$ holds on a Zariski dense set.  Then HNP for polynomials 
$f_1,\ldots,f_k \in \Q[\vec{x}][y_1,\ldots,y_n]$ can be written as
$\forall^* \boldsymbol x\, \exists \boldsymbol y \, \bigwedge_{i=1}^k f_i(\boldsymbol x,\boldsymbol y)=0$.

\subsection{Related work}
The $\HN$ problem is known to be polynomial-time equivalent to the \emph{tensor rank} problem; see, e.g.~\cite[Theorem~3]{shitov2016hard}.
A more general problem related to $\HN$ is the \emph{radical membership} problem.  Garg and Saxena~\cite{Garg020} study the complexity of radical membership parameterised by the transcendence degree of the polynomials in the system. 

We study the $\HN$ problem over algebraically closed fields.
The problem of determining whether a family of polynomials has real zero is known to lie in {\PSPACE}~\cite{canny} and has given rise to the complexity class $\exists \R$~\cite{SchaeferS17}.

We work with a formulation of $\HN$ in the Turing model of computation.  The problem is also naturally expressed in the Blum-Shub-Smale model, where it is the canonical complete problem for the analogue of {\NP} in that setting~\cite{burgisser2013completeness}.

\section{Overview}
\label{sec:background}
In this section we give a high-level overview of our main result and its proof. We also introduce some 
notation along the way. 
We recall that {\AM} is the class of languages with an interactive
proof that consist of the verifier sending a   random string, the prover responding with a string,
both  of  polynomial length, 
and where the decision to accept is obtained by applying a deterministic polynomial-time function
to the transcript.~\cite[Section~8.4]{arorabarak}
All formal statements of results and corresponding proofs can be respectively found in \Cref{sec:technical} and 
in the Appendix in the full version~\cite{manssour2024parametricversionhilbertnullstellensatz}.

\subsection{The \texorpdfstring{$\HNZ$}{HNP} problem.}
Let  $\vec{x} := (x_1,\ldots,x_m)$, and recall that $\overline{\Q(\vec{x})}$ denotes the algebraic closure of the  field of rational functions in variables $\vec{x}$ over $\Q$.
We investigate the complexity of determining whether the system of polynomial equations
\begin{equation}
\label{eq:alg_system}
f_1(y_1,\ldots,y_n) = 0,\; \ldots, \; f_k(y_1,\ldots,y_n) = 0
\end{equation}
with $f_i \in \Q(\vec{x})[y_1,\ldots,y_n]$ has a solution in $\overline{\Q(\vec{x})}^n$ . 
We call this the Parametric Hilbert Nullstellensatz problem, or $\HNZ$ for short.  Henceforth we assume that the polynomials $f_i$ have coefficients in $\Z[\vec{x}]$.  
This is without loss of generality since any system with 
coefficients in $\Q(\vec{x})$ can be transformed by scaling to one with coefficients in $\Z[\vec{x}]$ that has the same set of solution over $\overline{\Q(\vec{x})}$.
The conventions for representing 
such polynomials will be detailed in Section~\ref{sec:conventions}.

\subsection{Reduction from \texorpdfstring{$\HNZ$}{HNP} to \texorpdfstring{$\HN$}{HN}}
\label{sec:hn-number-theoretic-overview}
Let $\vec{x} := (x_1, \ldots, x_m)$ and $\vec{y}:= (y_1, \ldots, y_n)$.
Henceforth we fix an instance of $\HNZ$, comprising a collection $\mathcal{S}=\{f_1,\ldots,f_k\} 
\subseteq \Z[\vec{x}][\vec{y}]$ of polynomials for which we seek a common zero in $\overline{\Q(\vec{x})}$.
We study the satisfiability of the system $\mathcal S$ when specialised at a given tuple of values $\vec{\alpha} := (\alpha_1,\ldots,\alpha_m)$ in $\Z^m$. That is, we consider whether the set of polynomials $\mathcal S_{\vec{\alpha}} := \{ f_1(\vec{\alpha},\vec{y}),\ldots,f_k(\vec{\alpha},\vec{y}) \} \subseteq \mathbb{Z}[\vec{y}]$ has a common zero in $\overline{\Q}$.
Our aim is to show that with 
high probability over a random choice of $\vec{\alpha}$ in a suitable range,
the systems~$\mathcal S$ and $\mathcal S_{\vec{\alpha}}$ are equisatisfiable; thereby we obtain a randomised polynomial-time reduction of $\HNZ$ to $\HN$.

\paragraph{Unsatisfiable systems.}
Assume that $\mathcal{S}$ is not satisfiable in $\overline{\Q(\vec{x})}$. The weak version of Hilbert's Nullstellensatz implies that the ideal generated by $f_1,\ldots,f_k$ contains the constant~$1$. In particular, there exists a  non-zero polynomial $a\in \Z[\vec{x}]$ and polynomials $g_1,\ldots,g_k \in \Z[\vec{x}][\vec{y}]$ such that
\[ a = g_1f_1 + \cdots + g_k f_k \, . \]
Observe that if the above equation holds then the specialised system $\mathcal{S}_{\vec{\alpha}}$ can only have a solution if $a(\vec{\alpha})=0$. 
An effective parametric version of Hilbert's Nullstellensatz (see Theorem~\ref{th:parametric-effective-hn-multivar}) gives a degree bound on $a(\vec{x})$. We may thus apply the Polynomial Identity Lemma (see Lemma~\ref{thm:schwartz-zippel})  to deduce a bound on the probability that $a(\vec{\alpha}) = 0$, and hence an upper-bound on the probability that $\mathcal{S}_{\vec{\alpha}}$ admits a solution in $\overline{\Q}$, where $\vec{\alpha}$ is chosen uniformly at random in a fixed range; see \Cref{prop:unsat-spec-bound-multivar}.

\paragraph{Satisfiable systems.}
Assume, conversely, that the system $\mathcal{S}$ is satisfiable.
We first show that there exists a ``small'' solution $\vec{\beta} = (\beta_1,\ldots,\beta_n) \in \overline{\Q(\vec{x})}^n$ of~$\mathcal{S}$. In particular, in  \Cref{prop:sat-small-solution}, we
use results on quantifier elimination for algebraically closed fields to 
obtain a bound on the degree (in $\vec{x}$ and $\vec{y}$) of the defining polynomials of the $\beta_i$'s. 

Let $\K := \Q(\vec{x})(\beta_1,\ldots,\beta_n)$ be the finite extension of~$\Q(\vec{x})$ obtained by adjoining the solution~$\vec{\beta}$ to $\Q(\vec{x})$. Since the extension $\K$ is separable over $\Q(\vec{x})$, by the Primitive Element Theorem there exists  $\theta \in \overline{\Q(\vec{x})}$ such that  $\K = \Q(\vec{x})(\theta)$, i.e.,
every element in $\mathbb K$ can be written as a $\Q(\vec{x})$-linear combination of the powers $1,\theta,\ldots,\theta^{N-1}$, where $N$ is the degree of $\K$ over $\Q(\vec{x})$.
In particular, for all $i\in \{1,\ldots,n\}$ we have
\begin{equation}
    \label{eq:beta-prim-el}
    \beta_i = \frac{P_i(\theta)}{b}  \,  ,
\end{equation}
where $P_i(y) \in \Q[\vec{x}][y]$ and $b \in \Q[\vec{x}]$.

Denote by 
$m_\theta(\vec{x},y) \in \Q[\vec{x}][y]$ the minimal polynomial of $\theta$ over~$\Q(\vec{x})$.
In \Cref{prop:prim-element-multivar} we exhibit an upper bound on the degree of $m_\theta(\vec{x},y)$ in both $y$ and $\vec{x}$.  We use this bound to obtain, in turn, a degree bound on the polynomial 
$b$ in~\eqref{eq:beta-prim-el}.
Specifically, in \Cref{prop:prim-el-poly-bounds-solution}, we show that $b \in \Q[\vec{x}]$ has total degree singly exponential in the representation of the system $\mathcal S$. 

Given $\vec{\alpha} \in \Z^m$, we claim 
that the system~$\mathcal{S}_{\vec{\alpha}}$ remains satisfiable over $\overline{\Q}$ if 
$b(\vec{\alpha}) \neq 0$.
Indeed, if 
$b(\vec{\alpha}) \neq 0$,
we have a ring homomorphism
$\varphi_{\vec{\alpha}} \colon \Q[\vec{x},b(\vec{x})^{-1}][\theta] \to \overline{\Q}$ defined by
\[ \sum_{i = 0}^{N-1} q_j(\vec{x}) \theta^j \mapsto \sum_{j = 0}^{N-1} q_j(\vec{\alpha}) \omega^j \]
where $q_j(\vec{x}) \in \Q[\vec{x},b(\vec{x})^{-1}]$ 
and $\omega$ is a root of $m_\theta(\vec{\alpha},y) \in \overline{\Q}[y]$. 
Note that for $p\in \Z[\vec{x}]$ we have
$\varphi_{\vec{\alpha}}(p) = p(\vec{\alpha})$, i.e., the restriction of $\varphi_{\alpha}$ to $\Z[\vec{x}]$ is just evaluation at $\vec{\alpha}$.  It follows that the image of $\vec{\beta}$
under $\varphi_{\vec{\alpha}}$ 
is a solution of the specialised system
$\mathcal S_{\vec{\alpha}}$.
We can use the Polynomial Identity Lemma to upper bound the probability that $b(\vec{\alpha})=0$ (and hence that $\mathcal S_{\vec{\alpha}}$ is not satisfiable over~$\overline{\Q}$) for $\vec{\alpha}$ chosen uniformly at random from a fixed range.

\paragraph{Reduction to $\HN$.}
Finally, we combine the results on unsatisfiable and satisfiable systems to compute a bound $D$ such that if we choose $\alpha_1,\ldots,\alpha_m$ in $\{1,2,\ldots,D\}$ independently and uniformly at random, with high probability
$\mathcal S$ and $\mathcal S_{\vec{\alpha}}$ are equisatisfiable.
We thus obtain a randomised polynomial-time reduction from $\HNZ$ to $\HN$, placing $\HNZ$ in {\AM} assuming GRH.

\section{Detailed proofs}
\label{sec:technical}
\subsection{Conventions}
\label{sec:conventions}
Let $\vec{x}=(x_1,\ldots,x_m)$
and $\vec{y}=(y_1,\ldots,y_n)$ be tuples of variables.
In this section $\mathcal S=\{f_1,\ldots,f_k\}$ denotes a finite collection of polynomials in 
$\Z[\vec{x}][\vec{y}]$ given in sparse representation, i.e., as lists of monomials with all integer constants in binary.
We define the \emph{size} of $\mathcal S$
to be the maximum of the number $m$ of \emph{parameters}, number $n$ of \emph{variables}, number $k$ of polynomials, and the 
bit length of the integers appearing as coefficients or exponents in the description of the polynomials.
We further assume that the total degree of each polynomial $f_i$ in~$\vec{x}$ and $\vec{y}$ is at most~$2$ and that the integer coefficients are either $0$ or $1$.
 This assumption is without loss of generality since
 there is a straightforward polynomial-time transformation of an arbitrary system into an equisatisfiable system that satisfies our conventions (which proceeds by introducing new variables for subexpressions and using iterated squaring to handle exponentiation);
 see~\cite[Section~1.1]{KOIRAN1996273,KOIRAN-technical}.
In what follows, the size of the systems we consider is thus equal to $\max(m,n,k)$.

Recall from Section~\ref{sec:background} that for the given system $\mathcal{S}$, its specialisation at $\vec{\alpha}\in \Z^m$ is denoted $\mathcal S_{\vec{\alpha}}$.
We proceed to compute explicit bounds on the probability (over the choice of $\vec{\alpha}$ in a suitable range) that $\mathcal S$ and 
$\mathcal S_{\vec{\alpha}}$ are equisatisfiable, based on which we give a randomised 
polynomial-time reduction of $\HNZ$ to $\HN$.

\subsection{Unsatisfiable systems}
\label{sec:hn-number-theoretic-unsat}
The following is an effective version of the weak Hilbert Nullstellensatz for $\overline{\Q(\vec{x})}$ given in~\cite[Theorem~5 and Corollary~4.20]{d2013heights}.

\begin{theorem}[Effective Parametric Hilbert's Nullstellensatz]
	\label{th:parametric-effective-hn-multivar}
	Let $f_1,\ldots,f_k \in \Z[\vec{x}][y_1,\ldots,y_n]$ be a family of $k$ polynomials  of degree at most $2$ in $\vec{x}$ and~$\vec{y}$ that have no common zero in $\overline{\Q(\vec{x})}$.
	Then there exists $a \in \Z[\vec{x}]\setminus \{0\}$ and $g_1,\ldots,g_k \in \Z[\vec{x}][y_1,\ldots,y_n]$ such that
	\begin{equation}
	\label{eq:alg_system_hn}
	a = g_1f_1 + \cdots + g_k f_k,
    \end{equation}
	with
	\begin{itemize}
		\item $\deg_{\vec{y}}(g_if_i) \leq 2^k$,
		\item $\deg_{\vec{x}}(a), \deg_{\vec{x}}(g_if_i) \leq k2^k$.
	\end{itemize}
\end{theorem}

We use Theorem~\ref{th:parametric-effective-hn-multivar} to estimate the probability that an unsatisfiable system 
$\mathcal S$ remains unsatisfiable when specialised at $\vec{\alpha} \in \Z^m$.
To this end we will rely on the Polynomial Identity Lemma~\cite{DEMILLO1978193, schwartz, zippel}:

\begin{lemma}[Polynomial Identity Lemma]
\label{thm:schwartz-zippel}
Let $f \in K[x_1,\ldots,x_n]$ be a non-zero polynomial of total degree $d \geq 0$ over an integral domain $K$. Let $S$ be a finite subset of $K$ and let $r_1,\ldots, r_n$ be selected independently and uniformly at random from $S$. Then
\[ \Pr\left(f(r_1,\ldots, r_n) = 0\right) \leq \frac{d}{|S|}.\]
\end{lemma}

Recall that we use~$s$ to denote the size of our fixed system $\mathcal S$.  Then we have:

\begin{proposition}
\label{prop:unsat-spec-bound-multivar}
    Let $D \in \N$ be such that $D \geq s2^s$.
	If $\mathcal{S}$ has no solution in $\overline{\Q(\vec{x})}$, then for $\alpha_1,\ldots,\alpha_m$ chosen independently and uniformly at random from $\{1,2,\ldots,D\}$
 \[ \Pr\left(\mathcal{S}_{\vec{\alpha}} \text{ is satisfiable in } \overline{\Q}\right) \leq \frac{s2^s}{D}. \]
\end{proposition}

\begin{proof}
	If $\mathcal{S}$ has no solution in $\overline{\Q(\vec{x})}^n$ then by the weak Hilbert Nullstellensatz, the ideal generated by the polynomials in $\mathcal{S}$ contains the constant $1$. Furthermore, \Cref{th:parametric-effective-hn-multivar} asserts that there exists $a \in \Q[\vec{{x}}]$ of degree at most $s2^s$ that can be computed as a $\Z[\vec{x}]$-linear combination of the polynomials in $\mathcal{S}$, i.e., Equation~\eqref{eq:alg_system_hn} holds. For all $\vec{\alpha}$ such that $a(\vec{\alpha}) \neq 0$, the specialised system~$\mathcal{S}_{\vec{\alpha}}$  cannot be satisfiable in $\overline{\Q}$.
	That is, by the union bound
	 $$\Pr\left(\mathcal{S}_{\vec{\alpha}} \text{ is satisfiable in } \overline{\Q}\right) \leq \Pr\left(a(\alpha_1,\ldots,\alpha_r) = 0\right).$$  But $a$ has degree at most $s2^s$, and hence by \Cref{thm:schwartz-zippel} the bound follows.
\end{proof}

\subsection{Satisfiable systems}
\label{sec:hn-number-theoretic-sat}

\subsubsection{Degree bound on solutions}
\label{sec:sat-small-sol}

In \cite[Theorem~7]{KOIRAN-technical}, Koiran showed that a system of polynomials with integer coefficients
that has a common zero has one such whose
degree is bounded by a certain explicit 
function of the size of the system.
His proof relies on quantifier-elimination bounds for $\overline{\Q}$. We follow similar reasoning to obtain an upper bound on the degree of a solution to a satisfiable polynomial systems with coefficients in $\Z[\vec{x}]$. For the full proof, which is a close adaptation of \cite[Theorem~7]{KOIRAN-technical}; see \Cref{app:smallsol}.

\begin{restatable}{proposition}{smallsol}
 \label{prop:sat-small-solution}
	Let $\mathcal{P}$ be a system of $k$ polynomial equations in $n$ variables $\vec{y} = (y_1,\ldots,y_n)$ with coefficients in~$\Z[\vec{x}]$ of degree at most $2$ in $\vec{x}$ and $\vec{y}$.
	There exists an effective constant $c\in\N$ such that if $\mathcal{P}$ has a solution over $\overline{\Q(\vec{x})}$, then there exists a solution $\vec{\beta} = (\beta_1,\ldots,\beta_n)$ such that each component $\beta_i$ is a root of a polynomial of degree at most $2^{(n\log k)^{c}}$ with coefficients that are polynomials in $\Z[\vec{x}]$ of degree at most $2^{((n+m)\log k)^{c}}$ in $\vec{x}$.   
\end{restatable}

\subsubsection{Primitive element}
\label{sec:prim-el}

Assume that the system~$\mathcal{S}$ is satisfiable. Let $\vec{\beta} = (\beta_1,\ldots,\beta_n)$ be a solution of~$\mathcal{S}$  in $\overline{\Q(\vec{x})}^n$ satisfying the bounds stated in Proposition~\ref{prop:sat-small-solution}
and write $\K := \Q(\vec{x})(\beta_1,\ldots,\beta_n)$.
Denote by $\mathcal O_{\mathbb K}$ the subring of $\mathbb K$ comprising those elements that are \emph{integral} over $\Q[\vec{x}]$, that is, whose minimal polynomial over $\Q(\vec{x})$ is monic with coefficients in $\Q[\vec{x}]$.

Recall that we write $s$ to denote the size of the system $\mathcal S$. 
In the following proposition 
we compute a polynomial $d \in \Q[\vec{x}]$ such that $\tilde{\beta}_i := d\beta_i \in \mathcal{O}_{\K}$ for all  $i \in \{1,\ldots,n\}$ and then follow 
 a standard construction of a primitive element $\theta \in \mathcal{O}_{\K}$ as a $\Z$-linear combination of the $\tilde{\beta}_1,\ldots,\tilde{\beta}_n$. 
 (See \Cref{app:primel} for details.)

\begin{restatable}{proposition}{primel}
    \label{prop:prim-element-multivar}
	There exists a primitive element $\theta$ of $\K$ with monic minimal polynomial~$m_\theta(\vec{x},y) \in \Q[\vec{x}][y]$ of degree at most $2^{(s\log s)^{c}}$ in $y$, and with coefficients of degree at most $2^{(s\log s)^{c}}$  in $\vec{x}$, where $c$ is an effective constant.
\end{restatable}

We next show how to recover representations of 
the components $\beta_i$ of the solution of $\mathcal S$ in terms of the primitive element $\theta$.

\begin{proposition}
\label{prop:prim-el-poly-bounds-solution}
    There exists $b \in \Z[\vec{x}]$ of degree at most $2^{(s\log s)^c}$ in $\vec{x}$ and $P_1,\ldots,P_n \in \Z[\vec{x}][y]$
    such that for all $i \in \{1,\ldots,n\}$, $\beta_i = \frac{P_i(\theta)}{b}$, where $c$ is an effective constant.
\end{proposition}

\begin{proof}
By \Cref{prop:integral-denominator} in \Cref{app:denominators-disc} we have 
$\mathcal O_{\mathbb K} \subseteq \frac{1}{\mathrm{disc}(m_\theta)} \sum_{i=0}^{N-1} \Q[\vec{x}]\theta^i$, where
$\mathrm{disc}(m_\theta)  \in \Q[\vec{x}]$ is the discriminant of the minimal polynomial $m_\theta$ of $\theta$.  In general,
the discriminant of a polynomial of degree $N$ is a polynomial in its coefficients of total degree $2N-1$.
Thus, by
Proposition~\ref{prop:sat-small-solution}, $\mathrm{disc}(m_\theta)$ is a polynomial
in $\Q[\vec{x}]$ of total degree at most $(2N - 1)\cdot 2^{(s\log
s)^{c}} \leq 2^{(s\log s)^{c'}}$, where $c' > c$ is again an effective
absolute constant.  

Recall that
$\tilde{\beta}_i = d\beta_i \in \mathcal{O}_{\K}$ for all $i \in \{1,\ldots,n\}$ where $d
\in \Q[\vec{x}]$ of total degree at most $2^{(s\log s)^{c}}$ (see \Cref{app:primel}). We can thus write $\beta_i =
\frac{P_i(\theta)}{b}$, where $b := d\cdot \mathrm{disc}(m_\theta)$ is of total degree at
most $2^{(s\log s)^{c''}}$ for an effective absolute constant $c''$.
\end{proof}

\subsubsection{Specialisation}

We now show that if the 
denominator~$b$
 from \Cref{prop:prim-el-poly-bounds-solution} does not vanish at  $\vec{\alpha} \in \Z^m$ then the system $\mathcal{S}_\vec{\alpha}$ is satisfiable in $\overline{\Q}$. The degree bound 
  on $b$
 in combination with the Polynomial Identity Lemma allows us to compute a bound on the probability of $\mathcal{S}_{\vec{\alpha}}$ admitting a solution over $\overline{\Q}$ when $\vec{\alpha}$ is  chosen uniformly at random in a fixed range.

\begin{proposition}
\label{prop:sat-spec-bound-multivar}
    There exists an effectively computable constant $c$ such that for all $D \in \N$ with $D \geq 2^{(s\log s)^c}$, if $\mathcal{S}$ has a solution in $\overline{\Q(\vec{x})}$, then for $\alpha_1,\ldots,\alpha_m$ chosen independently and uniformly at random from $\{1,2,\ldots,D\}$
    \[ \Pr\left(\mathcal{S}_{\vec{\alpha}} \text{ is satisfiable in } \overline{\Q}\right) \geq 1 - \frac{2^{(s\log s)^c}}{D}. \]
\end{proposition}

\begin{proof}
Suppose $\mathcal{S}$ has a solution $\vec{\beta} = (\beta_1,\ldots,\beta_n)$ in $\overline{\Q(\vec{x})}^n$.
Recall we set $\theta$ to be a primitive element for $\K = \Q(\vec{x})(\beta_1,\ldots,\beta_n)$ that is integral over $\Q[\vec{x}]$ and denote by $m(\vec{x}, y) \in \Q[\vec{x}][y]$ its minimal polynomial.
 By \Cref{prop:prim-el-poly-bounds-solution}, there exist polynomials $P_1, \ldots, P_n \in \Z[\vec{x}][y]$ and $b \in \Z[\vec{x}]$ such that for all $\beta_i = \frac{P_i(\theta)}{b}$. Furthermore, the degree of $b$ is at most $2^{(s\log s)^{c}}$ where $c$ is an effective constant.

We first show that given $\vec{\alpha} \in \overline{\Q}^m$, the system $\mathcal{S}_{\vec{\alpha}}$ is satisfiable over $\overline{\Q}$ 
if $b$ does not
vanish on~$\vec{\alpha}$. 
For all $i \in \{1,\ldots,k\}$, let
    \begin{equation}
        \label{eq:univar-system}
        g_i(y) = b^{2} f_i\left(\frac{P_1(y)}{b},\ldots, \frac{P_n(y)}{b}\right).
    \end{equation}
   As discussed in the beginning of \Cref{sec:background}, we assume the degrees of the $f_i$'s in $y$ are at most $2$, thus it must be that $g_i \in \Z[\vec{x}]$. Furthermore, the polynomials $g_i$ must be multiples of $m$ since $m$ is irreducible and $g_i(\theta) = 0$.

    If $b$ does not 
    vanish on $\vec{\alpha}$, \eqref{eq:univar-system} must also hold in $\overline{\Q}$ when specialised at $\vec{x} = \vec{\alpha}$. This implies that if $\omega$ is a solution of $m(\vec{\alpha}, y)$ in $\overline{\Q}$, then 
    $\left(\frac{P_1(\vec{\alpha}, \omega)}{b(\vec{\alpha})},\ldots, \frac{P_n(\vec{\alpha}, \omega)}{b(\vec{\alpha})}\right)$ 
    is a solution of~$\mathcal{S}_{\vec{\alpha}}$ in $\overline{\Q}$. Note that such a solution $\omega$ of  $m(\vec{\alpha}, y) \in \Q[y]$  exists for all $\vec{\alpha}\in\overline{\Q}^m$ as the polynomial $m$ is monic in~$y$, hence $m(\vec{\alpha},y)$ is not constant.
    
Thus
\[ \Pr\left(\mathcal{S}_{\vec{\alpha}} \text{ is satisfiable in } \overline{\Q}\right) \geq \Pr\left(b(\vec{\alpha}) \neq 0\right) = 1 - \Pr\left(b(\vec{\alpha}) = 0\right),  \]
which by the Polynomial Identity Lemma is at least $1 - \frac{2^{(s\log s)^c}}{D}$.
\end{proof}

\subsection{Reduction to \texorpdfstring{$\HN$}{HN}.}
\label{sec:hn-number-theoretic-reduction}
We have just given estimates on the probability that the system $\mathcal{S}$ specialised at an  integer point $\vec{\alpha}$ randomly chosen in a fixed range, is satisfiable over $\overline{\Q}$. We will now use them to exhibit a randomised polynomial-time reduction of $\HNZ$ to $\HN$. This allows us to use Koiran's {\AM} protocol for $\HN$ to decide $\HNZ$ in {\AM} assuming GRH as well.

Recall that without loss of generality we may assume {\AM} protocols satisfy perfect completeness~\cite[Section~3]{FurerGMSZ89}. 
That is,
given a system of polynomial equations in $\Z[y_1,\ldots,y_n]$ that is satisfiable in $\overline{\Q}$, we suppose the {\AM} protocol for $\HN$ outputs ``satisfiable'' with probability~$1$, whereas given a system that is not satisfiable over $\overline{\Q}$, the probability of outputting ``satisfiable'' is at most $1/2$. We note that by standard amplification techniques, the error bound for false positives can be improved to $1/2^\ell$ for a fixed $\ell \in \N$.

\thmmain*

\begin{proof}
Let $D := 3\cdot 2^{(s\log s)^{c}}$ where $c$ is the effective constant from \Cref{prop:sat-spec-bound-multivar}.
We reduce our fixed instance of $\HNZ$ to an instance of $\HN$ by choosing the values $\alpha_1,\ldots,\alpha_m$ independently and uniformly at random from the set $\{1,2,\ldots,D\}$, and running Koiran's algorithm from~\cite{KOIRAN-technical} on $\mathcal{S}_\vec{\alpha}$. 
Let us now analyse the probabilistic correctness of our reduction.

Denote by $\mathcal{A}$ the event that Koiran's algorithm outputs ``satisfiable'', and by $\mathcal{A}^c$ its complement, i.e., the event that the algorithm outputs ``unsatisfiable''. Let $\mathcal{B}$ be the event that for $\alpha_1,\ldots,\alpha_m$ chosen uniformly at random from $\{1,2,\ldots,D\}$, the system $\mathcal{S}_{\vec{\alpha}}$ is satisfiable in $\overline{\Q}$, and let $\mathcal{B}^c$ be the complement of $\mathcal{B}$, i.e., the event that $\mathcal{S}_{\vec{\alpha}}$ is not satisfiable in $\overline{\Q}$.

Suppose that the system $\mathcal{S}$ is not satisfiable over $\overline{\Q(\vec{x})}$.
By \Cref{prop:unsat-spec-bound-multivar}, if we choose $\alpha_1,\ldots,\alpha_m$ uniformly at random in $\{1,2,\ldots, D\}$ the probability that $\mathcal{S}_\alpha$ is satisfiable over~$\overline{\Q}$ (the event $\mathcal{B}$) is at most $\frac{s2^s}{D}$. That is, since $D \geq 4s 2^s$, we have $\Pr\left(\mathcal{B}\right) \leq \frac{1}{4}$.

We recall that Koiran's algorithm admits perfect completeness, that is, $\Pr\left(\mathcal{A}^c\mid \mathcal{B}\right) = 0$. By amplifying the error probability for unsatisfiable instances and repeating his algorithm $\ell = 4$ times, we further have $\Pr\left(\mathcal{A} \mid \mathcal{B}^c\right) \leq \frac{1}{16}$.

We can now bound from below the probability of the algorithm outputting ``unsatisfiable'' as follows.
\begin{align*}
    \Pr\left(\mathcal{A}^c\right) 
    &= \Pr\left(\mathcal{B}\right)\cdot \Pr\left(\mathcal{A}^c \mid \mathcal{B}\right) + \Pr\left(\mathcal{B}^c\right)\cdot \Pr\left(\mathcal{A}^c \mid \mathcal{B}^c\right) \\
    &= \Pr\left(\mathcal{B}^c\right)\cdot \Pr\left(\mathcal{A}^c \mid \mathcal{B}^c\right) \\
    &\geq \Big(1 - \frac{1}{4} \Big) \cdot \Big(1 - \frac{1}{16} \Big)  = \frac{45}{64} \geq \frac{2}{3} \; ,
\end{align*}
which is the desired probability of correctness for unsatisfiable instances.	
	
Let us now suppose the system $\mathcal{S}$ is satisfiable over $\overline{\Q(\vec{x})}$.
By \Cref{prop:sat-spec-bound-multivar}, the probability that $\mathcal{S}_{\vec{\alpha}}$ is satisfiable over $\overline{\Q}$ when $\alpha_1,\ldots,\alpha_m$ are chosen uniformly at random in $\{1,2,\ldots,D\}$ (the event $\mathcal{B}$) is at least $\frac{2^{(s\log s)^{c}}}{D}$. 
In particular, our choice of $D = 3\cdot 2^{(s\log s)^{c}}$, implies that $\Pr\left(\mathcal{B}\right) \geq \frac{2}{3}$.

Recall that Koiran's algorithm admits perfect completeness, hence $\Pr\left(\mathcal{A}\mid \mathcal{B}\right) = 1$.

We can bound the probability of the algorithm outputting ``satisfiable'' as follows.
\begin{align*}
    \Pr\left(\mathcal{A}\right) 
     &= \Pr\left(\mathcal{C}\right)\cdot \Pr\left(\mathcal{B} \mid \mathcal{C}\right) + \Pr\left(\mathcal{C}^c\right)\cdot \Pr\left(\mathcal{B} \mid \mathcal{C}^c\right) \\
    &\geq \Pr\left(\mathcal{B}\right)\cdot\Pr\left(\mathcal{A} \mid \mathcal{B}\right) \\
    &\geq \frac{2}{3}\cdot 1 = \frac{2}{3}
\end{align*}
We have again obtained the required probability of correctness for satisfiable instances, which completes our randomised polynomial-time reduction.
Since {\AM} is closed under probabilistic reductions~\cite[Section~8.4.2]{arorabarak},
this proves the theorem.

\end{proof}

\paragraph{Acknowledgements}
Nikhil Balaji, Klara Nosan and Mahsa Shirmohammadi are supported by International
Emerging Actions grant (IEA’22). Rida Ait El Manssour and Mahsa Shirmohammadi are supported by
ANR grant VeSyAM (ANR-22-CE48-0005). James Worrell was supported by UKRI Frontier Research Grant EP/X033813/1.

\bibliographystyle{plain}
\bibliography{literature}

\begin{thebibliography}{10}

\bibitem{arorabarak}
Sanjeev Arora and Boaz Barak.
\newblock {\em Computational Complexity: A Modern Approach}.
\newblock Cambridge University Press, USA, 1st edition, 2009.

\bibitem{BanderierD15}
Cyril Banderier and Michael Drmota.
\newblock Formulae and asymptotics for coefficients of algebraic functions.
\newblock {\em Comb. Probab. Comput.}, 24(1):1--53, 2015.

\bibitem{burgisser2013completeness}
Peter B{\"u}rgisser.
\newblock {\em Completeness and reduction in algebraic complexity theory},
  volume~7.
\newblock Springer Science \& Business Media, 2013.

\bibitem{canny}
John Canny.
\newblock Some algebraic and geometric computations in pspace.
\newblock In {\em Proceedings of the twentieth annual ACM symposium on Theory
  of computing}, pages 460--467, 1988.

\bibitem{cohen2013course}
Henri Cohen.
\newblock {\em A Course in Computational Algebraic Number Theory}.
\newblock Graduate Texts in Mathematics. Springer Berlin Heidelberg, 2013.

\bibitem{cox2013ideals}
David Cox, John Little, and Donal {O'Shea}.
\newblock {\em Ideals, varieties, and algorithms: an introduction to
  computational algebraic geometry and commutative algebra}.
\newblock Springer Science \& Business Media, 2013.

\bibitem{DEMILLO1978193}
Richard~A. Demillo and Richard~J. Lipton.
\newblock A probabilistic remark on algebraic program testing.
\newblock {\em Information Processing Letters}, 7(4):193--195, 1978.

\bibitem{d2013heights}
Carlos d’Andrea, Teresa Krick, and Mart{\'i}n Sombra.
\newblock Heights of varieties in multiprojective spaces and arithmetic
  nullstellens{\"a}tze.
\newblock In {\em Annales scientifiques de l'{\'E}cole Normale Sup{\'e}rieure},
  volume~46, pages 549--627, 2013.

\bibitem{FITCHAS19901}
Noah Fitchas, André Galligo, and Jacques Morgenstern.
\newblock Precise sequential and parallel complexity bounds for quantifier
  elimination over algebraically closed fields.
\newblock {\em Journal of Pure and Applied Algebra}, 67(1):1--14, 1990.

\bibitem{frohlich1991algebraic}
Albrecht Fr{\"o}hlich and Martin~J Taylor.
\newblock {\em Algebraic number theory}.
\newblock Number~27. Cambridge University Press, 1991.

\bibitem{FurerGMSZ89}
Martin F{\"{u}}rer, Oded Goldreich, Yishay Mansour, Michael Sipser, and Stathis
  Zachos.
\newblock On completeness and soundness in interactive proof systems.
\newblock {\em Adv. Comput. Res.}, 5:429--442, 1989.

\bibitem{Garg020}
Abhibhav Garg and Nitin Saxena.
\newblock Special-case algorithms for blackbox radical membership,
  nullstellensatz and transcendence degree.
\newblock In Ioannis~Z. Emiris and Lihong Zhi, editors, {\em {ISSAC}'20:
  International Symposium on Symbolic and Algebraic Computation, Kalamata,
  Greece, July 20-23, 2020}, pages 186--193. {ACM}, 2020.

\bibitem{jel-effective}
Zbigniew Jelonek.
\newblock On the effective nullstellensatz.
\newblock {\em Inventiones mathematicae}, 162(1):1--17, 2005.

\bibitem{Koiran95-volume}
Pascal Koiran.
\newblock Approximating the volume of definable sets.
\newblock In {\em 36th Annual Symposium on Foundations of Computer Science,
  Milwaukee, Wisconsin, USA, 23-25 October 1995}, pages 134--141. {IEEE}
  Computer Society, 1995.

\bibitem{KOIRAN1996273}
Pascal Koiran.
\newblock Hilbert's nullstellensatz is in the polynomial hierarchy.
\newblock {\em Journal of Complexity}, 12(4):273--286, 1996.

\bibitem{KOIRAN-technical}
Pascal Koiran.
\newblock Hilbert's nullstellensatz is in the polynomial hierarchy.
\newblock Technical Report 96-27, DIMACS, 07 1996.

\bibitem{koiran-dim}
Pascal Koiran.
\newblock Randomized and deterministic algorithms for the dimension of
  algebraic varieties.
\newblock In {\em 38th Annual Symposium on Foundations of Computer Science,
  {FOCS}'97, Miami Beach, Florida, USA, October 19-22, 1997}, pages 36--45.
  {IEEE} Computer Society, 1997.

\bibitem{kollar}
J{\'a}nos Koll{\'a}r.
\newblock Sharp effective nullstellensatz.
\newblock {\em Journal of the American Mathematical Society}, 1(4):963--975,
  1988.

\bibitem{krick01}
Teresa Krick, Luis~Miguel Pardo, and Mart{\'i}n Sombra.
\newblock Sharp estimates for the arithmetic nullstellensatz.
\newblock 2001.

\bibitem{KuichSalomaa}
Werner Kuich and Arto Salomaa.
\newblock {\em Semirings, Automata, Languages}.
\newblock Monographs in Theoretical Computer Science. An EATCS Series. Springer
  Berlin Heidelberg, 1986.

\bibitem{kururPhd}
Piyush~P. Kurur.
\newblock {\em Complexity Upper Bounds using Permutation Group theory}.
\newblock PhD thesis, University of Madras, Chennai, 2006.

\bibitem{manssour2024parametricversionhilbertnullstellensatz}
Rida Ait~El Manssour, Nikhil Balaji, Klara Nosan, Mahsa Shirmohammadi, and
  James Worrell.
\newblock A parametric version of the hilbert nullstellensatz.
\newblock {\em arXiv preprint arXiv:2408.13027}, 2024.

\bibitem{mayr-meyer}
Ernst~W. Mayr and Albert~R. Meyer.
\newblock The complexity of the word problems for commutative semigroups and
  polynomial ideals.
\newblock {\em Advances in mathematics}, 46(3):305--329, 1982.

\bibitem{milneGalTheory}
James~S. Milne.
\newblock Fields and galois theory, 2021.
\newblock Lecture notes (version 5.0).

\bibitem{SchaeferS17}
Marcus Schaefer and Daniel Stefankovic.
\newblock Fixed points, nash equilibria, and the existential theory of the
  reals.
\newblock {\em Theory Comput. Syst.}, 60(2):172--193, 2017.

\bibitem{schwartz}
Jacob~T. Schwartz.
\newblock Fast probabilistic algorithms for verification of polynomial
  identities.
\newblock {\em J. ACM}, 27(4):701–717, 10 1980.

\bibitem{shitov2016hard}
Yaroslav Shitov.
\newblock How hard is the tensor rank?
\newblock {\em arXiv preprint arXiv:1611.01559}, 2016.

\bibitem{zippel}
Richard Zippel.
\newblock Probabilistic algorithms for sparse polynomials.
\newblock In {\em Proceedings of the International Symposiumon on Symbolic and
  Algebraic Computation}, EUROSAM '79, page 216–226, Berlin, Heidelberg,
  1979. Springer-Verlag.

\end{thebibliography}

\appendix
\section{Proof of Proposition~\ref{prop:sat-small-solution}}
\label{app:smallsol}

We work with the first-order theory of algebraically closed fields of characteristic zero. Let $K$ be a field, and denote by~$\overline{K}$ its algebraic closure. We consider the first-order language~$\mathcal{L}$ with constant symbols for all elements of $K$, function symbols $+, -, \cdot$, and the relation symbol $=$. Atomic formulas have the form $P(x_1,\ldots,x_n) = 0$, where $P \in K[x_1,\ldots,x_n]$.  We say that a formula~$\Phi$ is \emph{built over a set of polynomials $\mathcal{P}$} if every polynomial mentioned in~$\Phi$ lies in~$\mathcal{P}$.  It is well-known that the theory of algebraically
closed fields admits quantifier elimination. We use the
following quantitative formulation of quantifier elimination over $\overline{\Q(\vec{x})}$, which is a specialisation of~\cite[Theorem~2]{FITCHAS19901}, and its reformulation stated in~\cite[Theorem~6]{KOIRAN1996273}, to~$\overline{\Q(\vec{x})}$.

\begin{theorem}
\label{th:elim-multivar}
	Let $\mathcal{P}$ be a set of $k$ polynomials each of degree at most~$d$.
	Fix $Y=(y_1,\ldots,y_{k_1})$ and $Z=(z_1,\ldots,z_{k_2})$ to be tuples of first-order variables.
	Consider the formula 
	\[ \Phi(Y):= \exists Z \, \Psi(Y,Z) \, , \]
    where $\Psi(Y,Z)$ is a quantifier-free formula built over $\mathcal{P}$.
	There exists an equivalent quantifier-free formula $\Phi'(Y)$ that is built
	over a set of polynomials $\mathcal{Q}$ with degree bounded by 
	$2^{n^{O(1)}(\log kd)^{O(1)}}$.
	 The number of polynomials in $\mathcal{Q}$ is
	$O\big((kd)^{n^{O(1)}}\big)$.
	 
	 Moreover, when the coefficients of the polynomials in $\Psi$ are elements of $\Z[\vec{x}]$, and the combined degree of the polynomials in $\vec{x}, Y$ and $Z$ is bounded by $d$, then the coefficients of the polynomials in $\Phi'$ are elements of $\Z[\vec{x}]$ of degree at most
	$2^{(n+r)^{O(1)}(\log kd)^{O(1)}}$
	in~$\vec{x}$.
\end{theorem}

We now use the above result to show that polynomial systems with coefficients in $\Z[\vec{x}]$, when satisfiable, admit a bounded-degree algebraic solution.
We closely follow the proof \cite[Theorem~7]{KOIRAN-technical}, adapting it to $\overline{\Q(\vec{x})}$.

\smallsol*

\begin{proof}
We first prove the claim for the case when the system $\mathcal{P}$ has a finite number of solutions.
Let $S \subseteq \overline{\Q(\vec{x})}^n$ be the solution set of $\mathcal{P}$ and suppose it is finite. Write $S_i \subseteq \overline{\Q(\vec{x})}$ for the projection of~$S$ on the $i$th coordinate axis. By \Cref{th:elim-multivar}, there exists an effective constant $c\in\N$ such that~$S_i$ can be defined by a quantifier-free formula built over a set of polynomials $\mathcal{Q}$ having degree bounded by $2^{(n\log kd)^{c}}$ and coefficients  that are polynomials in $\Z[\vec{x}]$ of degree at most $2^{((n+m)\log kd)^{c}}$ in~$\vec{x}$.  If~$S$ is finite, then each $S_i$ is finite. Hence each element of $S_i$ is a root of some polynomial in~$\mathcal{Q}$.
Now given a solution $\vec{\beta} \in S$, its components are contained in $S_1,\ldots,S_n$, and thus are roots of polynomials in~$\mathcal{Q}$. 

Let us now assume that $\mathcal{P}$ has infinitely many solutions. We prove the claim by induction on~$n$.
Let $S \subseteq \overline{\Q(\vec{x})}^n$ be the solution set of $\mathcal{P}$. By the assumption on $\mathcal{P}$, at least one of the projections $S_i$ must be infinite.
Now $S_i$ is a constructible set, and its algebraic closure $\overline{S_i}$ is a whole line, hence its complement $\overline{\Q(\vec{x})} \setminus S_i$ must be finite. 
Furthermore, by \Cref{th:elim-multivar} the elements of $\overline{\Q(\vec{x})} \setminus S_i$ are chosen among the roots of $O\big(k^{n^{O(1)}}\big)$ polynomials of degree at most $2^{n^{O(1)}(\log kd)^{O(1)}}$.
Hence $|\overline{\Q(\vec{x})} \setminus S_i| < 2^{(n \log k)^{c}}$ for some effective constant $c \in\N$. This implies that there exists an integer $m \in S_i$ with $0 \leq m \leq 2^{(n \log k)^{c}}$. By substituting $y_i$ with $m$ in $\mathcal{P}$, we obtain a new satisfiable system in $n-1$ variables where the polynomials are of combined degree at most $2$ in $\vec{x}$ and $\vec{y}\setminus \{y_i\}$. By the induction hypothesis, the result follows.
\end{proof}

\section{Proof of Proposition~\ref{prop:prim-element-multivar}}
\label{app:primel}

Recall that an algebraic field extension $K/L$ is said to be \emph{separable} if for every $\alpha\in K$, the minimal polynomial of $\alpha$ over $L$ is a separable polynomial. All algebraic extensions of characteristic zero are separable. Separable extensions admit the following property.
\begin{theorem}[Primitive Element Theorem]
\label{th:primitive_element}
Let $K/L$ be a separable extension of finite degree. Then $K = L(\theta)$ for some $\theta \in K$; that is, the extension is simple and $\theta$ is a primitive element. 
\end{theorem}

The proof of the Primitive Element Theorem is constructive (see, for example, \cite[Theorem~4.1.8]{cohen2013course} or \cite[Theorem 5.1]{milneGalTheory}), and computes the primitive element~$\theta$ as a linear combination of the generators of the finite extension. That is, if $K := L(\alpha_1,\ldots,\alpha_k)$, then $\theta = \sum_{i=1}^k c_i \alpha_i$. The computation of $\theta$ is done inductively, constructing first a primitive element $\theta_2$ for $L(\alpha_1,\alpha_2)$, then $\theta_3$ for $L(\alpha_1, \alpha_2,\alpha_3)$, and so on until $\theta$ is obtained.  Furthermore, one can show that only finitely many combinations of the constants~$c_i$ fail to generate a primitive element for the field extension~$K$. In particular, if $L$ is an extension of $\Q$, the constants can be chosen in $\Z$, as summarised in the following lemma (which is a generalisation of \cite[Proposition 6.6]{kururPhd}).

\begin{lemma}
\label{lem:effective-primitive-element}
    Let $\alpha, \beta$ be algebraic elements of $\overline{\Q(\vec{x})}$ of respective degrees $\ell$ and~$m$ over $\Q(\vec{x})$. There exists an integer $c \in \{1,\ldots, \ell^2m^2 + 1\}$ such that $\alpha + c\beta$ is a primitive element for~$\Q(\vec{x})(\alpha,\beta)$.
\end{lemma}

As noted in \Cref{sec:hn-number-theoretic-sat},  given an element $\alpha \in K$, we can always find a polynomial $d\in\Q[\vec{x}]$ such that $\alpha d = \beta \in \mathcal{O}_K$.
 Indeed, let $f(y) \in \Q(\vec{x})[y]$ be the minimal (monic) polynomial of~$\alpha$ and choose $d$ to be the least common multiple of the denominators of the coefficients of $f(y)$. 
Then, since $f$ is monic,
\[ d^{\mathrm{deg} f} f\left(\frac{y}{d}\right) = g(y), \]
and $g(y) \in \Q[\vec{x}][y]$ is monic, with $\alpha d$ as a root. Thus $\alpha d \in \mathcal{O}_K$. Furthermore, $\alpha d$ and $\alpha$ have the same degree over $\Q(\vec{x})$.

Recall we write $s$ to denote the size of our fixed system.  By \Cref{prop:sat-small-solution}, the generators $\beta_i$ of the field $\K$ are algebraic over $\Q(\vec{x})$ of degree at most~$2^{(s\log s)^{c'}}$ for an effective constant $c' \in \N$.
By the argument above, for each $\beta_i$, there must exist a polynomial~$d_i \in \Q[\vec{x}]$ of degree at most $2^{(s\log s)^{c'}}$ in $\vec{x}$ such that $d_i\beta_i \in \mathcal{O}_{\K}$. Let $d = \lcm_{i=1}^{n} d_i$ and observe that its total degree is again bounded by $2^{(s\log s)^{c''}}$  for an effective constant $c'' > c' $.
Fix $\tilde{\beta}_i := d\beta_i \in \mathcal{O}_{\K}$ for all $i \in \{1,\ldots,n\}$ .
We obtain a primitive element $\theta \in \mathcal{O}_{\K}$ as a $\Z$-linear combination of the generators $\tilde{\beta}_1,\ldots,\tilde{\beta}_n$.

\primel*

\begin{proof}
	Denote by $D := 2^{(s\log s)^{c'}}$ the bound on the degrees of the $\beta_i$'s (and hence the $\tilde{\beta}_i$'s) given by \Cref{prop:sat-small-solution}.
	We follow the proof of the Primitive Element Theorem and construct the primitive element as
	$\theta = \sum_{i=1}^{n} c_i\tilde{\beta}_i$ where the $c_i \in \Z$ are of magnitude at most $D^{2n}+1$, and $\theta$ is of degree at most $D^n$.
	We prove the bound on the constants~$c_i$ by induction on $n$.
 
For $n = 2$, we construct a primitive element~$\theta_2 = \tilde{\beta}_1 + c_2\tilde{\beta}_2$ for the field $\Q(\vec{x})(\tilde{\beta}_1,\tilde{\beta_2)}$. By \Cref{lem:effective-primitive-element}, we can choose $c_2$ of magnitude at most $d^4 + 1$. To obtain the bound on the degree of $\theta_2$, denote by $p(y), q(y) \in \Q[\vec{x}](y)$ the respective minimal polynomials of $\tilde{\beta}_1$ and $\tilde{\beta}_2$ over $\Q(\vec{x})$. That is, we have $p(\tilde{\beta}_1) = q(\tilde{\beta}_2) = 0$. Notice that the polynomials $p(\theta_2 - cy)$ and $q(y)$ have a common root $\tilde{\beta}_2$.  Now recall that the resultant of two polynomials is a polynomial expression of their coefficients that is equal to zero if and only if the polynomials have a common root. That is, the resultant of $p(\theta_2 - cy)$ and $q(y)$ is a polynomial expression in $\theta_2$ equal to zero. Since the resultant is of degree at most~${D}^2$, $\theta_2$ must be of degree at most $D^2$ over $\Q(\vec{x})$.

Now let us assume that the bounds hold for $n - 1$, that is, the primitive element~$\theta_{n-1}$ of the field $\Q(\vec{x})(\tilde{\beta}_1,\ldots,\tilde{\beta}_{n-1})$ can be constructed as $\theta_{n-1} = \sum_{i=1}^{n-1} c_i\tilde{\beta}_i$ with $c_i$ integers of magnitude at most ${D}^{2(n-1)} + 1$, and that it is algebraic of degree at most ${D}^{n-1}$ over~$\Q(\vec{x})$.
By \Cref{lem:effective-primitive-element}, we can construct the primitive element of the field $\K = \Q(\vec{x})(\theta_{n-1},\tilde{\beta}_n)$ as a linear combination $\theta_{n-1} + c_n\tilde{\beta}_n$ where $c_n$ is an integer of magnitude at most $({D}^{n-1})^2D^2 + 1 = D^{2n} + 1$. The bound on the degree of $\theta$ again follows from the resultant argument applied to the minimal polynomials of $\theta_{n-1}$ and $\tilde{\beta}_n$ over $\Q(\vec{x})$.

	We have thus shown that the minimal polynomial $m_\theta(\vec{x},y)$ of $\theta$ is of degree at most $D^{2n}\leq 2^{(s\log s)^c}$ where $c$ is an effective constant.
	In remains to prove the bound on the coefficients of~$m_\theta$. 
	To this aim, 
	we construct a system of $k+1$ polynomial equations in $n+1$ variables with coefficients in $\Z[\vec{x}]$ such that one of its solutions will be $(\beta_1,\ldots,\beta_n,\theta)$.
	
	For all $i \in \{1,\ldots,k\}$, let $g_i(y_1,\ldots,y_n,z) = f_i(y_1,\ldots,y_n)$, where $f_i$ is as in $\mathcal{S}$.
	Define $$g_{k+1}(y_1,\ldots,y_n,z) := z - \sum_{i = 1}^{n} c_i d  y_i,$$
	where $d \in \Q[\vec{x}]$ is the polynomial such that $\tilde{\beta}_i = d \beta_i$ for all $i\in \{1,\ldots,n\}$.
	 Then the system
	\begin{equation}
		\label{eq:alg_system-prim-el}
		g_1(y_1,\ldots,y_n,z) = 0, \; \ldots, \; g_{k+1}(y_1,\ldots,y_n,z) = 0
	\end{equation}
 	is satisfiable and admits the claimed solution.
 	Furthermore, notice that the combined degree of the polynomial $g_{k+1}$ in $\vec{x}$, $\vec{y}$ and $z$ is at most $2^{(s\log s)^{c''}}$, where $c''$ is an effective constant. By using the same reasoning as in the beginning of \Cref{sec:background}, we can transform the system in~\eqref{eq:alg_system-prim-el} into a system of polynomials of total degree at most~$2$, of size polynomial in $s$.
 	We may now apply \Cref{prop:sat-small-solution} to this new system to assert that the coefficients of $m_\theta(\vec{x},y)$ are of degree at most $2^{(s\log s)^c}$ in $\vec{x}$. 
 	Moreover, since $\theta \in \mathcal{O}_{\K}$ by construction, its minimal polynomial is monic.
\end{proof}

\section{Structure of the ring of integers}
\label{app:denominators-disc}

\begin{proposition}
\label{prop:integral-denominator}
Let $\mathbb K=\Q(\vec{x})(\theta)$
where $\theta \in \mathcal O_{\mathbb K}$ has minimal polynomial
$m_\theta \in \Q[\vec{x}][y]$ over
$\Q(\vec{x})$ of degree $N$.  Then 
$\mathcal O_{\mathbb K} \subseteq \frac{1}{\mathrm{disc}(m_\theta)} \sum_{i=0}^{N-1} \Q[\vec{x}]\theta^i$.
\end{proposition}
\begin{proof}
Given $\alpha \in \mathcal{O}_{\mathbb K}$, 
write $\alpha =
\sum_{i=0}^{N-1} q_i \theta^i$, where $q_0,\ldots,q_{N-1} \in
\Q(\vec{x})$.  Let $\sigma_0,\ldots,\sigma_{N-1}$ be a list of the
monomorphisms from $\Q(\vec{x})$ to $\mathbb{K}$ that fix
$\Q(\vec{x})$.  Applying these monomorphisms to the previous equation
gives $\sigma_j(\alpha) = \sum_{i=0}^{N-1} q_i \sigma_j(\theta^i)$ for
$j=0,\ldots,N-1$.  Solving this system of linear equations for
$q_0,\ldots,q_{N-1}$ using Cramer's rule we obtain
\begin{gather}
q_i = \frac{\det(D_i)}{\det(D)} =
\frac{\det(D_i)\det(D)}{\det(D)^2} \, ,
\label{eq:disc}
\end{gather}
where $D = (\sigma_j(\theta^i))_{i,j}$ and $D_i$ is the matrix
obtained from $D$ by replacing the $i$th column with the vector 
$(\sigma_0(\alpha),\ldots,\sigma_{N-1}(\alpha))^\top$.

The denominator $\det(D)^2$ on the right-hand side of~\eqref{eq:disc} is the
discriminant of $m_\theta$; see,
e.g. \cite[p. 15, Equation~(1.25.a)]{frohlich1991algebraic}. Moreover the 
numerator $\det(D_i)\det(D)$ on the right-hand side of ~\eqref{eq:disc} lies in $\Q(\vec{x})$ and is integral
over $\Q[\vec{x}]$.  Since $\Q[\vec{x}]$ is
integrally closed we thus have that $\det(D_i)\det(D) \in \Q[\vec{x}]$.
\end{proof}

\end{document}